\documentclass[a4paper]{article}
\usepackage{pdfpages}
\usepackage[T1]{fontenc}

\usepackage{fullpage}

\usepackage[
  colorlinks,
  linkcolor = blue,
  citecolor = blue,
  urlcolor = blue]{hyperref}

\usepackage{amsmath,amsthm,amssymb}
\usepackage{cleveref}

\newtheorem{theorem}{Theorem}[section]
\newtheorem{lemma}[theorem]{Lemma}
\newtheorem{proposition}[theorem]{Proposition}
\newtheorem{definition}[theorem]{Definition}
\newtheorem{corollary}[theorem]{Corollary}
\newtheorem{example}{Example}

\crefname{theorem}{Theorem}{Theorems}
\crefname{lemma}{Lemma}{Lemmas}
\crefname{proposition}{Proposition}{Propositions}
\crefname{definition}{Definition}{Definitions}
\crefname{corollary}{Corollary}{Corollaries}
\crefname{example}{Example}{Examples}
\crefname{section}{Section}{Sections}
\crefname{appendix}{Appendix}{Appendices}
\crefname{table}{Table}{Tables}

\usepackage{dsfont} \usepackage{mathrsfs} \usepackage{microtype}

\usepackage{thmtools,thm-restate}

\usepackage{xcolor}

\renewcommand{\P}[2][]{{\textnormal{Pr}_{#1}}{\left[#2\right]}}
\newcommand{\indic}[1]{\delta\sof[\big]{#1}} \newcommand{\wtM}{\widetilde{M}}
\newcommand{\wtN}{\widetilde{N}}

\newcommand{\bra}[1]{\langle{#1}|}
\newcommand{\ket}[1]{|{#1}\rangle}
\newcommand{\braket}[2]{\langle{#1}|{#2}\rangle}
\newcommand{\ketbra}[2]{\ket{#1}\bra{#2}}
\newcommand{\proj}[1]{\ketbra{#1}{#1}}
\newcommand{\kb}[1]{\proj{#1}}

\newcommand{\C}{\mathbb{C}}

\newcommand{\eqdef}{:=}

\newcommand{\POVM}[1]{\mathrm{M}(#1)} \newcommand{\PM}[1]{\mathrm{PM}(#1)} \newcommand{\D}[1]{\mathrm{D}(#1)} \newcommand{\U}[1]{\mathrm{U}(#1)} 

\newcommand{\one}{\mathds{1}}

\newcommand{\pg}[5][]{\omega^{#1}_{\textnormal{#2}}(#3|#4)_{#5}}

\newcommand{\reg}[1]{\mathsf{#1}}
\newcommand{\A}{\reg{A}}
\newcommand{\B}{\reg{B}}
\newcommand{\X}{\reg{X}}
\newcommand{\E}{\reg{E}}

\newcommand{\AB}{\reg{AB}}
\newcommand{\XA}{\reg{XA}}
\newcommand{\XB}{\reg{XB}}
\newcommand{\XAB}{\reg{XAB}}
\renewcommand{\a}{\reg{A'}}
\renewcommand{\b}{\reg{B'}}
\newcommand{\ab}{{\a\b}}
\newcommand{\Aa}{{\A\a}}
\newcommand{\Bb}{{\B\b}}

\newcommand{\hilb}[1]{\mathcal{#1}}
\newcommand{\HA}{\hilb{A}}
\newcommand{\HB}{\hilb{B}}

\newcommand{\HX}{\hilb{X}}
\newcommand{\Ha}{\hilb{A'}}
\newcommand{\Hb}{\hilb{B'}}
\newcommand{\da}{d}
\newcommand{\db}{d}

\newcommand{\fset}[1]{\mathscr{#1}}
\newcommand{\SA}{\fset{A}}
\newcommand{\SB}{\fset{B}}
\newcommand{\SM}{\fset{M}}
\newcommand{\SV}{\fset{V}}

\newcommand{\SE}{\fset{E}}

\newcommand{\SX}{\fset{X}}

\usepackage{mathtools}
\mathtoolsset{centercolon}
\DeclarePairedDelimiter{\set}{\lbrace}{\rbrace}
\DeclarePairedDelimiter{\abs}{\lvert}{\rvert}
\DeclarePairedDelimiter{\card}{\lvert}{\rvert} \DeclarePairedDelimiter{\norm}{\lVert}{\rVert}
\DeclarePairedDelimiter{\of}{\lparen}{\rparen}
\DeclarePairedDelimiter{\pr}{\lparen}{\rparen} \DeclarePairedDelimiter{\sof}{\lbrack}{\rbrack}

\newcommand{\mx}[1]{\begin{pmatrix}#1\end{pmatrix}}

\newcommand{\x}{\otimes}
\newcommand{\tp}{^{\mathsf{T}}} \newcommand{\ct}{^{\dagger}}
\DeclareMathOperator{\tr}{tr}

\title{Local simultaneous state discrimination}

\usepackage{authblk}
\newcommand{\email}[1]{\thanks{\href{mailto:#1}{#1}}}
\newcommand{\aff}[1]{\textit{\normalsize#1}}

\author[1]{Christian Majenz\email{chmaj@dtu.dk}}
\author[2]{Maris Ozols\email{marozols@gmail.com}}
\author[3]{Christian Schaffner\email{c.schaffner@uva.nl}}
\author[4]{Mehrdad Tahmasbi\email{mehrdad@cwi.nl}}

\affil[1]{\aff{Department of Applied Mathematics and Computer Science, Technical University of Denmark}}
\affil[2]{\aff{Institute for Logic, Language, and Computation, Korteweg-de Vries Institute for Mathematics, and Institute for Theoretical Physics, University of Amsterdam and QuSoft}}
\affil[3]{\aff{Institute for Logic, Language, and Computation, University of Amsterdam and QuSoft}}
\affil[4]{\aff{Centrum Wiskunde \& Informatica and QuSoft}}

\begin{document}
\maketitle

\begin{abstract}
	Quantum state discrimination is one of the most fundamental problems studied in quantum information theory. Applications range from channel coding to metrology and cryptography. In this work, we introduce a new variant of this task: Local Simultaneous State Discrimination (LSSD). While previous distributed variants of the discrimination problem always allowed some communication between the parties to come up with a joint answer, the parties in LSSD cannot communicate and have to simultaneously answer correctly. This simultaneity implies, e.g., that for classical states, the problem does not trivialize to a non-distributed distinguishing task. While interesting in its own right, this problem also arises in quantum cryptography.

	 After introducing the problem, we give a number of characterization results. We give examples showing that i) the optimal strategy for local discrimination need not coincide with the optimal strategy for LSSD, even for classical states, ii) an additional entangled resource can increase the optimal success probability in LSSD, and iii) stronger-than-quantum non-signalling resources can allow for a higher success probability in some cases, compared to strategies using entanglement. Finally, we show that finding the optimal strategy in (classical) 3-party LSSD is NP-hard.
\end{abstract}

\section{Introduction}
Discriminating between a known set of quantum states is a well-studied and fundamental problem in quantum information theory, with a vast range of applications ranging from cryptography and quantum computing to quantum information and metrology~\cite{Bae_Kwek_2015}. A referee randomly picks a quantum state from a known set of states and sends it to Alice who tries to determine which state was sent to her. An interesting extension of the problem is \emph{distributed} state discrimination where the states to be distinguished are bi-partite and Alice gets to examine register $\A$ and Bob register $\B$.
In the context of nonlocality, the most commonly considered scenario is \emph{LOCC} where Alice and Bob are allowed to use \emph{local operations and classical communication} in the discrimination process \cite{LOCC}.
For example, any orthonormal set of product states can be prepared by local operations and discriminated by a global one, however discriminating them with only local operations is generally not possible, even when classical communication between parties is allowed \cite{Nonlocality,Framework}.
In the LOCC setting, the discrimination task does not become more demanding by asking Alice and Bob to answer correctly \emph{simultaneously} since the result can be communicated between the parties.

Surprisingly, the more restricted scenario where Alice and Bob can only use \emph{local operations (LO)} without any classical communication has received only little attention in the published literature so far, see below for related work. In this scenario, asking both Alice and Bob to succeed simultaneously makes the task strictly more difficult compared with the case when at least one of the players should succeed. We call the resulting task \emph{local simultaneous state discrimination (LSSD)}.

While LSSD is certainly interesting in its own right, one concrete motivation --- in fact, our original motivation --- comes from quantum cryptography. Here, one line of work has studied \emph{unclonable cryptography}~\cite{Wiesner83,BB84,Gottesman03,Aaronson09,broadbent2019uncloneable,ALLZZ20,CLLZ21,MST21}. An unclonable cryptographic scheme is a scheme where a certain asset (like a token, message or functionality) is encrypted in a way that makes it impossible to copy. Such features are clearly impossible to achieve with purely classical means, and constructions make crucial use of the so-called \emph{quantum no-cloning principle} that states that quantum information, in general, cannot be copied. The general idea of using the no-cloning principle dates back to Wiesner \cite{Wiesner83} who proposed a \emph{quantum money} scheme where banknotes are quantum states, preventing copying. Later, quantum copy protection \cite{Aaronson09,ALLZZ20,ALP20,CMP20} and unclonable encryption \cite{broadbent2019uncloneable} were introduced, which provide more sophisticated assets in an unclonable way. Strengthening the standard encryption security notion of indistinguishability to \emph{indistinguishable unclonability} \cite{broadbent2019uncloneable} yields a security game that requires the adversary to perform LSSD.

Another motivation comes from the foundations of quantum mechanics. Quantum non-locality is a well-studied fundamental feature of quantum theory which has been key to charting the foundations of quantum physics. In particular, the characterization of non-local quantum correlations, both mathematically and operationally, constitutes a decades-old challenge, partially addressed by an impressive body of research (see, e.g., \cite{BellNonlocality} and references therein). This work establishes LSSD as a new and natural member of the zoo of operational problems (like non-local games and zero-error communication settings) where the non-local nature of quantum correlations can provide an advantage over strategies restricted to purely classical means, and stronger-than-quantum non-local correlations (so-called non-signaling boxes) can provide an additional advantage.

\subsection{Our contributions}
In this work, we define and study the problem of \emph{Local Simultaneous State Discrimination (LSSD)} which can be formalized by a tripartite cqq-state $\rho_{\XAB} = \sum_x P(x) \proj{x}_\X \x \rho_\AB^x$, where the referee's register $\X$ is classical and $\rho_{\AB}^x$ are arbitrary two-partite quantum states. Alice and Bob act locally on their respective registers $\A$ and $\B$ to produce guesses $x_A$ and $x_B$. They win the LSSD game if and only if both their guesses correctly identify the value $x$ of the classical register $\X$, i.e., $x = x_A = x_B$. As in non-local games, we can define optimal guessing probabilities by considering strategies for Alice and Bob that use different kinds of resources, namely: 1) shared randomness, 2) additional quantum entanglement, 3) non-signaling correlations. A priori, it is entirely unclear whether these extra resources allow Alice and Bob to increase their simultaneous guessing probability. The LSSD problem can also be naturally extended to more than two simultaneously distinguishing parties.

After setting the stage with these definitions, we provide a number of results for the LSSD problem where $\rho_{\XAB}$ is fully classical, i.e., Alice and Bob receive classical inputs $a,b$, correlated with the referee's $x$ according to a joint distribution $P_{\XAB}$. Our first result, \cref{prop:binary-input-bounds}, establishes that the three simultaneous guessing probabilities coincide if $x,a,b$ are all bits. Additionally, if only $a, b$ are bits, we prove a simple upper bound on the guessing probability with non-local correlations. In contrast, as our main contribution, we provide in \cref{thm:separation} a simple distribution $P_{\XAB}$ for which the three simultaneous guessing probabilities defined above are strictly separated from each other. Hereby, we establish that as for non-local games, having entangled strategies is (in general) strictly more powerful than shared randomness (which in turn is easily seen to be useless, as for non-local games). Also, having stronger non-signaling strategies (using, e.g., a Popescu-Rorlich box \cite{Popescu_Rohrlich_1994}) can be strictly more powerful than entanglement in LSSD.
Finally, in \cref{sec:hardness}, we study the computational complexity of finding optimal simultaneous guessing strategies by investigating (again fully classical) problem instances naturally defined based on $r$-partite hypergraphs. By establishing a connection between simultaneous guessing and finding a maximum matching in 3-partite hypergraphs, we show that finding an optimal classical strategy for the three-party LSSD problem is NP-hard.

\subsection{Related work}
Earlier work by Buscemi~\cite{Buscemi_2012} studied a very general classof distributed tasks called ``semi-quantum'' non-local games where a referee picks from a fixed set a bi-partite quantum state and sends the registers as questions to two players Alice and Bob, and their answers are classical bitstrings. A subclass of such games, namely quantum \emph{XOR games} have been studied in-depth by Regev and Vidick~\cite{Regev_Vidick_2012}. The restriction is that the players' answers are classical bits of which the referee only takes into account their XOR when computing the winning predicate. Our LSSD scenario is a similar subclass of semi-quantum games, where instead of the XOR condition, the players simultaneously have to guess the referee's choice. It is a very interesting open problem to investigate whether some of the results from quantum XOR games carry over to the LSSD setting. For instance, does there exist a family of games that can only be won optimally with an ever-increasing amount of entanglement?

Another notion of \emph{extended non-local games} has been defined and investigated by Russo~\cite{Russo_2017}. In extended non-local games, the referee, Alice and Bob share a quantum state, but the referee's questions and player's answers remain classical. However, the winning predicate is computed by a measurement of the referee. This setting ties in well with monogamy-of-entanglement games~\cite{MonogamyGame}, and it is shown in~\cite{Russo_2017} that some of the results~\cite{Regev_Vidick_2012} from quantum XOR games carry over to this setting. The main difference to our LSSD problem is that the initial quantum state is part of the players' strategy, and not prepared by the referee.

Another line of related work~\cite{Matthews_Wehner_Winter_2009,Lancien_Winter_2013,Lami2018} studies the relation between various distinguishability norms with the goal of maximising the so-called data hiding ratio, i.e.~how much worse restricted sets of measurements (such as local ones) perform in the task of state discrimination versus global measurements. In their setting, the ``local operations'' performed by the players can still be post-processed by the referee (akin to some form of communication), whereas in our LSSD setting, the players simultaneously have to guess the referee's input using only local measurements. This crucial difference is the reason why we observe interesting separations between the guessing probabilities already for the discrimination of fully classical states. When classical post-processing by a referee is allowed, the players can simply forward their classical inputs to the referee. Therefore, interesting effects in that setting only occur when distinguishing quantum inputs.

Very recent work in this line by Corrêa, Lami and Palazuelos~\cite{correa2021maximal} is also concerned with optimal local discrimination. By a clever combination of previous results about data hiding and the noncommutative Grothendieck's theorem, the authors show that the ratio between the optimal global distinguishing measurement between two states and the optimal local measurement is at most $2 \sqrt{2} d$ where $d$ is the local dimension of Alice an Bob's system. Due to the classical post-processing by the referee, their results cannot easily be translated into our LSSD setting.

During the preparation of this manuscript, we have become aware of independent unpublished work by Chitambar and Man\v cinska~\cite{CM21} that also studies the LSSD problem for two bipartite quantum states that are in tensor product. This setting can be seen as a quantum version of our \cref{ex:1} below. It shows the same ``two-regime behavior'', where depending on a parameter, it is better to use the locally optimal discrimination strategy in one regime, whereas in the other regime, it is better for the players to correlate their errors.

\subsection{Open problems}
We believe that LSSD is a fascinating new problem in quantum information processing, as there are many associated open questions. Our results in this article are exclusively\footnote{except \cref{ex:ex2}, which we import from \cite{MST21}}  concerned with the case where the referee uses classical states. How do the different success probabilities behave when distinguishing actual quantum states? Are there dimension constraints like in our \cref{prop:binary-input-bounds} under which the classical and quantum values coincide?

As mentioned above, can the results about quantum XOR games from~\cite{Regev_Vidick_2012} be ported to LSSD? Does there exist a family of games that can only be won optimally with an ever-increasing amount of entanglement? Can we find efficiently computable lower or upper bounds on the various success probabilities?

While we establish the NP hardness of finding optimal classical distinguishing strategies for three parties, it is natural to ask whether the two-party LSSD problem is already hard.

In terms of applications, we suggest to establish more links with uncloneable encryption and possibly with position-based cryptography.

\subsection{Notation}

We will denote by $\indic{\cdot}$ the indicator function that evaluates to one when its argument is true and to zero otherwise. We will use $\SX, \SA, \SB$, respectively, to denote the finite sets from which the inputs to the referee, Alice, and Bob are drawn.
Their joint input is described by a probability distribution $P_\XAB$ on $\SX \times \SA \times \SB$, wher the system $\X$ belongs to the referee while $\A$ and $\B$ belong to Alice and Bob, respectively.
The input and output sets will often be of the form
$[d] \eqdef \set{0,\dotsc,d-1}$,
for some integer $d \geq 1$.

When Alice and Bob's inputs are quantum, the overall input is a classical-quantum-quantum (cqq) state $\rho_\XAB$ where the classical register $\X$ belongs to the referee while the quantum registers $\A$ and $\B$ belong to Alice and Bob, respectively.
We will denote the finite-dimensional complex Euclidean spaces underlying these registers by $\HX = \C^\SX$, $\HA = \C^\SA$, and $\HB = \C^\SB$.

A \emph{quantum state} on $\C^d$ is a $d \times d$ positive semi-definite matrix of unit trace, i.e., $\rho \in \C^{d \times d}$ such that $\rho \succeq 0$ and $\tr \rho = 1$.
We denote the set of all quantum states on $\C^d$ by $\D{\C^d}$.
Operations on quantum states are described by \emph{unitary} matrices, i.e., $U \in \C^{d \times d}$ such that $U\ct U = \one$ where $\one$ is the identity matrix.
We denote the set of all unitaries on $\C^d$ by $\U{\C^d}$.

An $n$-outcome \emph{measurement} or POVM on $\C^d$ is a collection of $n$ positive semi-definite $d \times d$ matrices that sum to identity.
We will denote a measurement by $M = \set{M_1, \dotsc, M_n}$ where $M_i \succeq 0$ and $\sum_{i=1}^n M_i = \one$.
We denote the set of all $n$-outcome measurements on $\C^d$ by $\POVM{\C^d}$ (since the outcome set is always clear from the context, we do not specify it).
If $M_i^2 = M_i$ for all $i = 1, \dotsc, n$, we call the measurement \emph{projective}.
We denote the set of all $n$-outcome projective measurements on $\C^d$ by $\PM{\C^d}$.

\section{Local simultaneous state discrimination (LSSD) problem}
\label{sec:prob-form}

A referee prepares a tripartite system $\XAB$ in a cqq state
\begin{equation}
  \rho_{\XAB} = \sum_{x \in \SX} P_\X(x) \kb{x}_\X \otimes \rho_{\AB}^x
\end{equation}
and passes the $\A$ and $\B$ subsystems to two distant parties, Alice and Bob, respectively, while keeping the system $\X$. Alice and Bob know the state $\rho_{\XAB}$ and might share some resources (as will be precisely quantified later) prior to receiving their states, but no communication is allowed between them afterwards. Based on their received states and pre-shared resources, Alice and Bob output guesses $x_A$ and $x_B$, respectively, to the referee. They win if both guesses are correct, i.e., $x = x_A = x_B$, and they aim at maximizing their probability of winning. 

Most of our results are concerned with the case where $\rho_{\XAB}$ is completely classical, i.e., there exist orthonormal bases $\set{\ket{a}: a \in \SA}$ and $\set{\ket{b}: b \in \SB}$ for $\HA$ and $\HB$, respectively, that are independent of $x \in \SX$, and probability distributions $P_{\AB}^x$ over $\SA \times \SB$ such that
\begin{equation}
  \rho_{\AB}^x = \sum_{\substack{a \in \SA \\ b \in \SB}} P_{\AB}^x(a,b) \proj{a}_\A \x \proj{b}_\B.
\end{equation}

\paragraph{Classical Strategies.}
In this case, there are no additional resources available to Alice and Bob beyond their received state.\footnote{One can equivalently define classical strategies when only shared randomness is allowed between Alice and Bob. However, for the same reason as in non-local games, this purely classical resource does not help, as Alice and Bob could fix their randomness to a realization conditioned on which their probability of winning is maximized.}
The optimal probability of simultaneously guessing $x$ correctly is
\begin{align}
  \pg{c}{\X}{\A;\B}{\rho} \eqdef
  \sup_{\substack{M \in \POVM{\HA} \\ N \in \POVM{\HB}}}
  \sum_{x \in \SX}
  P_\X(x) \tr \sof[\big]{\rho_{\AB}^x \pr{M_x \x N_x}}.
\end{align}
When $\rho_{\XAB}$ is classical and described by a probability distribution $P_{\XAB}$, we can rewrite the optimal probability of winning as
\begin{align}
  \pg{c}{\X}{\A;\B}{P}
  &= \max_{\substack{Q_{\X_a|\A}\\ Q_{\X_b|\B}}} \sum_{\substack{x\in \SX\\a \in \SA , b \in \SB}} P_{\XAB}(x,a,b) Q_{\X_a|\A}(x_a|a) Q_{\X_b|\B}(x_b|b)\\
  &\stackrel{(1)}{=} \max_{f,g} \sum_{\substack{x\in \SX\\a \in \SA , b \in \SB}} P_{\XAB}(x,a,b) \indic{f(a) = g(b) = x},
\end{align}
where the first maximum is taken over all conditional probability distributions $Q_{\X_a|\A}$ and $Q_{\X_b|\B}$, the second maximum is taken over all functions $f: \SA \to \SX$ and $g: \SB \to \SX$, and $(1)$ follows since Alice and Bob can condition any local randomness on the realization that maximizes their probability of winning.

\paragraph{Quantum Strategies.}
In this case, Alice and Bob can share an entangled state prior to receiving their inputs. Let $\Ha = \Hb = \C^d$ be two complex Euclidean spaces of dimension $d$.
Alice and Bob first jointly prepare a quantum state $\sigma_\ab$ on $\Ha \x \Hb$, after which Alice and Bob keep systems $\a$ and $\b$, respectively.
After receiving their inputs, Alice and Bob determine their output by measuring the registers $\Aa$ and $\Bb$ with local measurements $M$ and $N$, respectively (this is the most general strategy because no communication is allowed).

When the local dimensions of the shared entangled state $\sigma_\ab$ are limited to $d$ for both parties, the optimal probability of winning is
\begin{align}
  \pg[d]{q}{\X}{\A;\B}{\rho} \eqdef
  \sup_{\substack{\sigma_\ab \in \D{\C^\da \x \C^\db}}}
  \sup_{\substack{M \in \POVM{\HA \x \C^\da} \\ N \in \POVM{\HB \x \C^\db}}}
  \sum_{x \in \SX} P_\X(x)
  \tr \sof[\big]{\of{\rho_{\AB}^x \x \sigma_\ab} \of{M_x \x N_x}}.
  \label{eq:pd}
\end{align}
When the dimensions of $\a$ and $\b$ are not limited, the optimal winning probability is
\begin{align}
	\pg{q}{\X}{\A;\B}{\rho} \eqdef
  \sup_{d \geq 1} \pg[d]{q}{\X}{\A;\B}{\rho}.
  \label{eq:pq def}
\end{align}
When $\rho_{\XAB}$ is classical and described by a probability distribution $P_{\XAB}$, we can simplify \cref{eq:pd} as follows:
\begin{align}
  \pg[d]{q}{\X}{\A;\B}{P}
 &= \sup_{\sigma_\ab \in \D{\C^\da \x \C^\db}}
    \sup_{\substack{M: \SA \to \POVM{\C^{\da}} \\ N: \SB \to \POVM{\C^{\db}}}}
    \sum_{\substack{x\in \SX \\ a \in \SA , b \in \SB}} P_{\XAB}(x, a, b)
    \tr \sof[\big]{\sigma_\ab \of[\big]{M_x(a) \otimes N_x(b)}} \\
 &= \sup_{\substack{M: \SA \to \POVM{\C^{\da}} \\ N: \SB \to \POVM{\C^{\db}}}}
    \norm[\bigg]{\sum_{\substack{x\in \SX \\ a \in \SA , b \in \SB}} P_{\XAB}(x,a,b) M_x(a) \x N_x(b)},
    \label{eq:pq for distribution}
\end{align}
where $M$ and $N$
are collections of measurements,
i.e., for every input $a \in \SA$ and $b \in \SB$,
we have that
$M(a) = \set{M_x(a) : x \in \SX}$ and
$N(b) = \set{N_x(b) : x \in \SX}$
are measurements on $\C^d$ with outcomes in $\SX$.
We show in \cref{cor:projective} that the optimization in $\pg{q}{\X}{\A;\B}{P}$ can be restricted to projective measurements.

\paragraph{No-signaling Strategies.}
We define no-signaling strategies only when $\rho_{\XAB}$ is classical and described by a probability distribution $P_{\XAB}$. Given classical inputs $a \in \SA$ and $b \in \SB$ for Alice and Bob, respectively, they output their estimates $x_A$ and $x_B$ of $x \in \SX$ according to a conditional probability distribution $Q_{\X_A\X_B|\AB}$ on $\SX \times \SX \times \SA \times \SB$ satisfying
\begin{align}
  \forall x_B, a, a', b: \quad
    \sum_{x_A \in \SX} Q_{\X_A\X_B|\AB}(x_A, x_B | a, b)
 &= \sum_{x_A \in \SX} Q_{\X_A\X_B|\AB}(x_A, x_B | a', b), \label{eq:ns1} \\
  \forall x_A, a, b, b': \quad
    \sum_{x_B \in \SX} Q_{\X_A\X_B|\AB}(x_A, x_B | a, b)
 &= \sum_{x_B \in \SX} Q_{\X_A\X_B|\AB}(x_A, x_B | a, b'). \label{eq:ns2}
\end{align}
An optimal no-signaling strategy succeeds with probability
\begin{align}
  \pg{ns}{\X}{\A;\B}{P} \eqdef
  \sup_{Q_{\X_A\X_B|\AB}}
  \sum_{\substack{x\in \SX \\ a \in \SA , b \in \SB}}
  P_{\XAB}(x,a,b)
  Q_{\X_A\X_B|\AB}(x,x|a,b).
  \label{eq:pns def}
\end{align}

\subsection{Examples}

We discuss here two examples of LSSD games. The first example highlights particular features of LSSD such as the optimal local strategies are not necessarily  optimal for simultaneous guessing, or the optimal guessing probability for product distributions is not the product of  the optimal guessing probability of distributions in general. The second example is related to applications of LSSD to quantum cryptography.

\begin{example} \label{ex:1}
Let $X$, $Y$, and $Z$ be independent binary random variables such that $\P{X=1} = 1/2$, $\P{Y=1} = \P{Z=1} = \alpha$ for some $0\leq \alpha \leq 1/2$. We also set $A \eqdef X\oplus Y$ and $B\eqdef X \oplus Z$ and denote the joint probability mass function of $(X, A, B)$ by $P^\alpha_{\XAB}$. In other words, $A$ and $B$ are independent noisy versions of the uniform bit $X$. Consider the problem of simultaneously guessing $X$ from $A$ and $B$. When $1-\frac{1}{\sqrt{2}}<\alpha < \frac12$, both parties always output $0$ regardless of their inputs, which is a correct guess of $X$ with probability $\frac12$. When  $0\leq \alpha  \leq 1-\frac{1}{\sqrt{2}}$, Alice and Bob estimate $X$ as $A$ and $B$, respectively, which are simultaneously correct when $Y = Z = 0$, an event that has probability $(1-\alpha )^2$. By a brute-force check, one finds that the aforementioned strategies are optimal without any extra resources and therefore
\begin{align}
    \pg{c}{\X}{\A;\B}{P^\alpha} =
    \begin{cases} \frac12 \quad & 1-\frac{1}{\sqrt{2}} \leq \alpha \leq \frac12,\\
    (1-\alpha)^2 \quad & 0\leq \alpha \leq 1-\frac{1}{\sqrt{2}}. \end{cases}
\end{align}
Note that when $1-\frac{1}{\sqrt{2}} \leq \alpha  \leq \frac12$, optimal local estimators of $X$ are not optimal for simultaneous guessing of $X$. We later show in \cref{prop:binary-input-bounds} that when all $X, A, B$ are binary, $\pg{c}{\X}{\A;\B}{P^\alpha }= \pg{q}{\X}{\A;\B}{P^\alpha } =\pg{ns}{\X}{\A;\B}{P^\alpha}$.

As a next observation, we set $\alpha := 1-\frac{1}{\sqrt{2}}$ and let $(X', A', B')$ be  an independent copy of $(X, A, B)$. We consider the simultaneous guessing  of $(X, X')$ from $(A, A')$ and $(B, B')$, and define a strategy as follows: both Alice and Bob output $(1, 1)$ if their input bits are $(1, 1)$ and output $(0, 0)$ otherwise. The probability of simultaneously guessing correctly is
\begin{align}
  \frac{1}{4} (1-\alpha ^2)^2 + \frac{1}{4} (1-\alpha )^4 \approx 0.271447.
\end{align}
Hence, $\pg{c}{\X\X'}{\A\A';\B\B'}{P^\alpha \times P^\alpha} > \pg{c}{\X}{\A;\B}{P^\alpha}\pg{c}{\X'}{\A';\B'}{P^\alpha}$ while $(X, A, B)$ and $(X', A', B')$ are independent. Because $\pg{c}{\X}{\A;\B}{P^\alpha }= \pg{q}{\X}{\A;\B}{P^\alpha} =\pg{ns}{\X}{\A;\B}{P^\alpha}$, we also have
\begin{align}
\pg{q}{\X\X'}{\A\A'; \B\B'}{P^\alpha\times P^\alpha} &> \pg{q}{\X}{\A;\B}{ P^\alpha}\pg{q}{\X'}{\A';\B'}{ P^\alpha}, \\
\pg{ns}{\X\X'}{\A\A'; \B\B'}{P^\alpha\times P^\alpha} &> \pg{ns}{\X}{\A;\B}{ P^\alpha}\pg{ns}{\X'}{\A';\B'}{P^\alpha}.
\end{align}
\end{example}
\begin{example}\label{ex:ex2}
Let $\HA = \HB = \C^3$ with an orthonormal basis $\set{\ket{0}, \ket{1}, \ket{\bot}}$ and let $\ket{\phi^x}_{\A\B}\eqdef \frac{1}{\sqrt{2}}\pr{\ket{x}\otimes \ket{\bot} + \ket{\bot} \otimes \ket{x} }$ for $x\in[2]$. We also set $\rho_{\X\A\B} \eqdef \frac12 \sum_{x\in[2]} \kb{x}_{\X} \otimes \kb{\phi^x}_{\A\B}$. The authors of \cite{MST21} showed that $\pg{c}{\X}{\A;\B}{\rho} \geq \frac{9}{16}$ and used this fact to prove impossibility of uncloneable encryption, as defined in \cite{MST21}, using pure states as ciphertext.
\end{example}

\section{Strict quantum and no-signaling separations for LSSD}

Our main result is the following theorem that gives a simple example of an LSSD problem for which the guessing probabilities for players with different types of shared resources are all distinct.
Namely,
$\pg{c}{\X}{\A;\B}{} < \pg{q}{\X}{\A;\B}{} < \pg{ns}{\X}{\A;\B}{}$.

\begin{theorem}\label{thm:separation}
Let $\SX = \set{0,1,2}$ and $\SA = \SB = \set{0,1}$, and let $P_{\XAB}$ be the uniform distribution over $\set{(0,1,0), (0,1,1), (1,0,0), (1,1,0), (2,0,1)}$. Then
\begin{align}
  \pg{c}{\X}{\A;\B}{P} &= 2/5 = 0.4, \label{eq:pc} \\
  \pg{q}{\X}{\A;\B}{P} =
  \pg[2]{q}{\X}{\A;\B}{P} &= \frac{16+\sqrt{13}}{45} \approx 0.435679, \label{eq:pq} \\
  \pg{ns}{\X}{\A;\B}{P} &= 1/2 = 0.5. \label{eq:pns}
\end{align}
\end{theorem}

Our proof relies on the following characterization of the classical and no-signaling guessing probabilities $\pg{c}{\X}{\A;\B}{P}$ and $\pg{ns}{\X}{\A;\B}{P}$ when $\abs{\SA} = \abs{\SB} = 2$ (see \cref{sec:pf-binary-output-pg} for proof).

\begin{restatable}{lemma}{pguess}\label{lm:binary-output-pg}
Let $P_{\XAB}$ be a probability distribution over $\SX \times \SA \times \SB$ with $\SA = \SB = \set{0,1}$ and $\SX = [d]$, $d \geq 2$.
The classical and no-signaling winning probabilities for $P_\XAB$ are given by
\begin{align}
  \pg{c}{\X}{\A;\B}{P}
 &= \max_{\substack{s,t \in \SX \\ s \neq t}}
    \max
    \set[\Big]{
      P_\X(s),
      P_{\XAB}(s, 0, 0) + P_{\XAB}(t, 1, 1),
      P_{\XAB}(s, 0, 1) + P_{\XAB}(t, 1, 0)
    }, \label{eq:cprob} \\
  \pg{ns}{\X}{\A;\B}{P}
 &= \max
    \set[\Big]{
      \pg{c}{\X}{\A;\B}{P},
      \max_{k \in \set{2, \dotsc, d}}
      \max_{f, g}
      \sum_{\substack{x\in\SX\\a\in\SA,b\in\SB}} P_\XAB(x,a,b)
      Q^k_{\X_A\X_B|\AB} \of[\big]{f(x,a),g(x,b)|a,b}
    },
    \label{eq:nsprob}
\end{align}
where the final maximization in \cref{eq:nsprob} is over all functions
$f: \SX \times \SA \to \SX$ and $g: \SX \times \SB \to \SX$ such that $f(\cdot,a), g(\cdot,b): \SX \to \SX$ are permutations for every $a \in \SA$ and $b \in \SB$,
and the conditional probability distribution $Q^k_{\X_A\X_B|\AB}$ on $\SX \times \SX \times \SA \times \SB$ is given by
\begin{equation}
  Q^k_{\X_A\X_B|\AB}(x_A,x_B|a,b) :=
  \begin{cases}
    \frac{1}{k} & \text{if $x_A,x_B \in [k]$ and $(x_A - x_B) \bmod k = ab$}, \\
    0 & \text{otherwise}.
  \end{cases}
  \label{eq:Qk}
\end{equation}
\end{restatable}

\begin{proof}[Proof (of \cref{thm:separation})]
The given distribution $P_{\XAB}$ has $P_\X(0) = P_\X(1) = 2/5$, $P_\X(2) = 1/5$, and $P_{\XAB}(x,a,b) \leq 1/5$ for all $x,a,b$. \Cref{eq:pc} then follows by applying \cref{lm:binary-output-pg}.
An explicit strategy achieving success probability $2/5$ is when both parties ignore their inputs and always output~$0$.

\begin{table}[!ht]
\centering
\begin{tabular}{c|ccc}
  $x$ & 0 & 1 & 2 \\
  \hline
  $f(x,0)$ & 2 & \textbf{1} & \textbf{0} \\
  $f(x,1)$ & \textbf{0} & \textbf{1} & 2 \\
  $g(x,0)$ & \textbf{0} & \textbf{1} & 2 \\
  $g(x,1)$ & \textbf{1} & 2 & \textbf{0}
\end{tabular}
\qquad\qquad
\begin{tabular}{ccc|c|c|c}
  $x$ & $a$ & $b$ & $ab$ & $f(x,a)$ & $g(x,b)$ \\
  \hline
  0 & 1 & 0 & 0 & 0 & 0 \\
  0 & 1 & 1 & 1 & 0 & 1 \\
  1 & 0 & 0 & 0 & 1 & 1 \\
  1 & 1 & 0 & 0 & 1 & 1 \\
  2 & 0 & 1 & 0 & 0 & 0
\end{tabular}
\caption{(Left)~An optimal choice of functions $f$ and $g$ for no-signaling strategies, see \cref{eq:nsprob}.
(Right)~We verify that for any $(x,a,b)$ with $P_{\XAB}(x) > 0$, $f(x,a), g(x,b) \in \set{0,1}$ (in bold) and $f(x,a) \oplus g(x,b) = ab$, hence this choice is compatible with \cref{eq:Qk} when $k = 2$.}
\label{tab:fg}
\end{table}

Next, let us prove \cref{eq:pns}.
Since $|\SX| = 3$, we only need to consider $k = 2$ and $k = 3$ in \cref{eq:nsprob} of \cref{lm:binary-output-pg}.
Note from \cref{eq:Qk} that $Q^k_{\X_A\X_B|\AB} \of{x_A,x_B|a,b} \leq \frac{1}{k}$ for any $x_A,x_B,a,b$,
so the corresponding term in \cref{eq:nsprob} is at most
\begin{align}
  \sum_{\substack{x\in\SX\\a\in\SA,b\in\SB}} P_\XAB(x,a,b)
  Q^k_{\X_A\X_B|\AB} \of[\big]{f(x,a),g(x,b)|a,b}
  \leq \frac{1}{k} \sum_{\substack{x\in\SX\\a\in\SA,b\in\SB}} P_\XAB(x,a,b)
  = \frac{1}{k}.
  \label{eq:1 over k}
\end{align}
If $k = 2$ and we choose $f, g: [3] \times [2] \to [3]$ according to \cref{tab:fg} then, for all $(x,a,b)$ with $P_{\XAB}(x,a,b) > 0$, we have $f(x,a), g(x,b) \in \set{0,1}$ and $f(x,a) \oplus g(x,b) = ab$, so the inequality in \cref{eq:1 over k} becomes tight.
According to \cref{eq:nsprob}, this lower bounds the success probability by $1/2$.
Since $k = 3$ can lower bound it by at most $1/3$, we do not need to consider this case.
Thus, according to \cref{lm:binary-output-pg}, $\pg{ns}{\X}{\A;\B}{P} = \max \set{2/5, 1/2} = 1/2$ which proves \cref{eq:pns}.

It remains to prove \cref{eq:pq}.
Let us denote the claimed optimal quantum value in \cref{eq:pq} by
\begin{equation}
  t_* := \frac{16+\sqrt{13}}{45}.
\end{equation}
We will first settle the case when the local dimension of the shared entangled state is $d = 2$, i.e., each party has a single qubit, and then reduce the general case of $d \geq 2$ to this one.

Towards establishing \cref{eq:pq}, let us first prove that $\pg[2]{q}{\X}{\A;\B}{P} \geq t_*$.
Alice and Bob can achieve the value $t_*$ by using the following strategy.
Their shared two-qubit state is
\begin{align}
  \ket{\sigma}_\ab &:= s_+ \ket{00}_\ab + s_- \ket{11}_\ab, &
  s_\pm := \sqrt{\frac{1}{2} \pm \frac{1}{78} \sqrt{715 - 182 \sqrt{13}}}.
\end{align}
To describe their measurements, we denote the qubit state at angle $\theta$ and the corresponding projector by
\begin{align}
  \ket{\psi(\theta)}
&:= \cos \theta \, \ket{0}
  + \sin \theta \, \ket{1}
  = \mx{\cos \theta \\ \sin \theta}, &
  \Pi(\theta)
&:= \proj{\psi(\theta)}
  = \mx{
      \cos^2 \theta & \cos \theta \sin \theta \\
      \sin \theta \cos \theta & \sin^2 \theta
    }.
  \label{eq:Pi}
\end{align}
Depending on their local inputs $a,b \in \set{0,1}$, Alice and Bob apply the projective measurements
$M(a) := \set{M_0(a), M_1(a), M_2(a)}$ and
$N(b) := \set{N_0(b), N_1(b), N_2(b)}$
given in \cref{tab:povms}.

\begin{table}[!ht]
\centering
\begin{tabular}{c|ccc}
  $x$ & 0 & 1 & 2\\
  \hline
  $M_x(0)$ & 0 & $\Pi(\alpha_0)$ & $\one - \Pi(\alpha_0)$ \\
  $M_x(1)$ & $\Pi(\alpha_1)$ & $\one - \Pi(\alpha_1)$ & 0 \\
  $N_x(0)$ & $\Pi(\beta_0)$ & $\one - \Pi(\beta_0)$ & 0 \\
  $N_x(1)$ & $\Pi(\beta_1)$ & 0 & $\one - \Pi(\beta_1)$
\end{tabular}
\caption{Measurements for Alice and Bob's quantum strategies. The projector $\Pi(\theta)$ is defined in \cref{eq:Pi} and their angles are given in \cref{eq:angles}.}
\label{tab:povms}
\end{table}

For each measurement, one of their operators is $0$ while the other two are of the form $\Pi(\theta)$ and $\one - \Pi(\theta)$, for some angles $\theta \in [-\pi/2,\pi/2]$.
The angles used in \cref{tab:povms} are chosen as follows:
\begin{align}
  \of{\alpha_0, \alpha_1, \beta_0, \beta_1}
  &:= \of{-\theta_1, \theta_2, \tfrac{\pi}{2} - \theta_2, \theta_1}, &
  \theta_1 &:= \tfrac{1}{4} \arccos\of*{\tfrac{ 121 + 52 \sqrt{13}}{477}}, &
  \theta_2 &:= \tfrac{1}{4} \arccos\of*{\tfrac{-431 +  4 \sqrt{13}}{477}}.
  \label{eq:angles}
\end{align}
The angles $\theta_1$ and $\theta_2$ satisfy
$\cos(4 \theta_1) = 12 + 13 \cos(4 \theta_2)$
and have the following explicit cosines:
\begin{align}
  \cos \theta_1 &= \sqrt{\frac{1}{318} \of[\bigg]{159 + \sqrt{689 \of[\big]{23 + 2 \sqrt{13}}}}}, &
  \cos \theta_2 &= \sqrt{\frac{1}{318} \of[\bigg]{159 + \sqrt{ 53 \of[\big]{23 + 2 \sqrt{13}}}}}.
\end{align}
Using a computer algebra system, one can verify that
\begin{equation}
  \bra{\sigma}_\ab
  \of[\bigg]{\sum_{\substack{x\in\SX\\a\in\SA,b\in\SB}} P_\XAB(x,a,b) M_x(a) \x N_x(b) }
  \ket{\sigma}_\ab
  = \frac{16 + \sqrt{13}}{45}
  = t_*.
\end{equation}
In fact, $\ket{\sigma}_\ab$ is the principal eigenvector of the above operator.\footnote{Indeed, one can check that its eigenvalues are
$\frac{16 + \sqrt{13}}{45}$,
$\frac{25 + \sqrt{13}}{90}$,
$\frac{ 7 + \sqrt{13}}{45}$,
$\frac{19 - 5 \sqrt{13}}{90}$.}

Next, let us prove that the above strategy is optimal if the shared entangled state has local dimension $d = 2$ and Alice and Bob use only projective measurements (we will later reduce the case of general measurements in any finite dimension $d$ to this).
For now, our goal is to show that
\begin{equation}
  \sup_{\substack{\Pi: \SA \to \PM{\C^2} \\ \Sigma: \SB \to \PM{\C^2}}}
  \norm[\Big]{\sum_{\substack{x\in\SX\\a\in\SA,b\in\SB}} P_{\XAB}(x,a,b) \Pi_x(a) \x \Sigma_x(b)}
  \leq t_*.
\end{equation}
First, by \cref{prop:zero-povm} we can assume that
\begin{align}
  \Pi_0(0) = \Pi_2(1) = \Sigma_2(0) = \Sigma_1(1) = 0
\end{align}
since Alice should not guess $0$ if $a=0$ and $2$ if $a=1$,
and Bob should not guess $2$ if $b=0$ and $1$ if $b=1$.
The remaining operators form two $2$-outcome projective measurements for each party:
\begin{align}
  \Pi(0) &= \set{\Pi_1(0), \Pi_2(0)}, &
  \Pi(1) &= \set{\Pi_0(1), \Pi_1(1)}, &
  \Sigma(0) &= \set{\Sigma_0(0), \Sigma_1(0)}, &
  \Sigma(1) &= \set{\Sigma_0(1), \Sigma_2(1)}.
\end{align}
To simplify notation, let us set
$(A_0, A_1, B_0, B_1) := \of[\big]{\Pi_0(0), \Pi_0(1), \Sigma_0(0), \Sigma_0(1)}$
so that
\begin{align}
  \Pi(0) &= \set{A_0, A_0^\perp}, &
  \Pi(1) &= \set{A_1, A_1^\perp}, &
  \Sigma(0) &= \set{B_0, B_0^\perp}, &
  \Sigma(1) &= \set{B_1, B_1^\perp}
  \label{eq:measurements}
\end{align}
where
$A_i^\perp := \one - A_i$ and
$B_i^\perp := \one - B_i$.
Our matrix of interest is then
\begin{align}
  \Omega &:=
  \sum_{\substack{x\in\SX\\a\in\SA,b\in\SB}} P_\XAB(x,a,b) \Pi_x(a) \x \Sigma_x(b)\\
 &= \frac{1}{5}
    \of[\big]{
      \Pi_0(1) \x \Sigma_0(0)
    + \Pi_0(1) \x \Sigma_0(1)
    + \Pi_1(0) \x \Sigma_1(0)
    + \Pi_1(1) \x \Sigma_1(0)
    + \Pi_2(0) \x \Sigma_2(1)
    } \\
 &= \frac{1}{5}
    \of[\big]{
      A_1 \x B_0
    + A_1 \x B_1
    + A_0 \x B_0^\perp
    + A_1^\perp \x B_0^\perp
    + A_0^\perp \x B_1^\perp
    }.
  \label{eq:AB}
\end{align}

We see from \cref{eq:measurements} that if any of the remaining Alice's measurement operators is $0$ then all her operators commute.
By \cref{lem:comm-meas} their winning probability cannot exceed the classical value
$\pg{c}{\X}{\A;\B}{P} = 2/5$.
Hence, all remaining Alice's measurement operators are rank-$1$, and similarly for Bob.

By applying a local unitary change of basis on Alice and Bob's systems, we can assume without loss of generality that, for some angles $\alpha, \beta \in [0,2\pi]$,
\begin{align}
  A_0 &= \mx{1&0\\0&0}, &
  A_1 &= \Pi\of*{\frac{\alpha}{2}}, &
  B_0 &= \mx{1&0\\0&0}, &
  B_1 &= \Pi\of*{\frac{\pi-\beta}{2}},
\end{align}
where $\Pi(\theta)$ is the projector defined in \cref{eq:Pi}.
With this choice, $\Omega$ from \cref{eq:AB} can be written as
\begin{equation}
  \Omega = \mx{
   -(a+1) (b-3) & (a+1) \sqrt{1-b^2} & -\sqrt{1-a^2} (b-3) & \sqrt{1-a^2} \sqrt{1-b^2} \\
    (a+1) \sqrt{1-b^2} & ab-a+b+7 & \sqrt{1-a^2} \sqrt{1-b^2} & \sqrt{1-a^2} (b-1) \\
   -\sqrt{1-a^2} (b-3) & \sqrt{1-a^2} \sqrt{1-b^2} & ab-3a+b+5 & -(a+1) \sqrt{1-b^2} \\
    \sqrt{1-a^2} \sqrt{1-b^2} & \sqrt{1-a^2} (b-1) & -(a+1) \sqrt{1-b^2} & -ab-b+a+5
  }
\end{equation}
where $a := \cos \alpha$ and $b := \cos \beta$.
Our goal is to show that $\norm{\Omega} \leq t_*$ over all $a,b \in [-1,1]$.
Using a computer algebra system we find that the characteristic polynomial of $\Omega$ in variable $t$ is
\begin{equation}
  f(t,a,b)
  = t^4 - t^3
  + \frac{32 + (1+a)(1+b)}{100} t^2
  - \frac{16 + 3 (1+a)(1+b)}{500} t
  + \frac{(1+a)(1+b) \of[\big]{4 - (1-a)(1-b)}}{5000}.
\end{equation}
Since the largest eigenvalue of $\Omega$ is equal to the largest root of $f$,
our goal is to show that $f$ has no roots $t > t_*$.
In \cref{lem:sos} in \cref{apx:sos} we find an exact sum of squares decomposition for $f$ which shows that
$f(t,a,b) > 0$ for any $t > t_*$ and $a,b \in [-1,1]$.
This implies that $f$ has no roots larger than $t_*$.

It remains to show that $\pg{q}{\X}{\A;\B}{P} \leq t_*$.
We will do this by reducing a general strategy to the above $d = 2$ problem.
Let us fix any dimension $d \geq 2$ and consider arbitrary local quantum strategies for Alice and Bob.
They are based on a shared state $\ket{\sigma}_\ab \in \C^\da \x \C^\db$ and collections of measurements $M: \SA \to \POVM{\C^\da}$ and $N: \SB \to \POVM{\C^\db}$.
By invoking \cref{prop:zero-povm} and then \cref{cor:projective} we can reduce $M$ and $N$ to two two-outcome projective measurements that look the same as in \cref{eq:measurements}, except that $A_i$ and $B_i$ are projectors in some finite-dimensional space $\C^{d'}$ where $d' \leq d \max \set{|\SA|, |\SB|}$.

For now, let us focus just on Alice's measurements.
They are fully parameterized by two projectors, $A_0$ and $A_1$.
By Jordan's Lemma~\cite{Jordan} (also known as \emph{CS decomposition}~\cite{CS}), there is a unitary change of basis on Alice's system that simultaneously block-diagonalizes $A_0$ and $A_1$:
\begin{align}
  A_0 &= \mx{
    \bigoplus_{j=1}^k \mx{1&0\\0&0} \\
    & \one \\
    && \one \\
    &&& 0 \\
    &&&& 0
  }, &
  A_1 &= \mx{
    \bigoplus_{j=1}^k \Pi(\theta_j) \\
    & \one \\
    && 0 \\
    &&& \one \\
    &&&& 0
  }.
\end{align}
Here the first $k$ blocks are of size $2 \times 2$ and contain rank-$1$ projectors onto $1$-dimensional subspaces at angle $\theta_j$ between them, see \cref{eq:Pi}.
The remaining blocks are $1 \times 1$ and contain values $(1,1)$, $(1,0)$, $(0,1)$, and $(0,0)$
(the number of times each pair occurs is determined by the sizes of the identity and all-zeroes matrices).
Note that $A_0^\perp$ and $A_1^\perp$ have similar block decompositions in the same basis.

We are interested in the largest eigenvalue of $\Omega$ defined in \cref{eq:AB}.
Since all Alice's projectors are block-diagonal, $\Omega$ is also block-diagonal (each Alice's block gets tensored by Bob's operator).
Since the largest eigenvalue of $\Omega$ must occur in one of these blocks, Alice might as well restrict her strategy to this single block.
Since each of her blocks has size at most two, her strategy does not require more than two dimensions.
By a similar argument, Bob's system can also be reduced to two dimensions.
Since we already analyzed strategies based on orthogonal measurements on a shared state with local dimension two, the same upper bound $t_*$ also applies to the general case.
\end{proof}

In the following proposition, we show that the example presented in \cref{thm:separation} is the ``smallest'' example illustrating a separation between the $\pg{c}{\X}{\A;\B}{P}$ and $\pg{ns}{\X}{\A;\B}{P}$ in a sense that when $\SA$ and $\SB$ have cardinality two, three is the minimum cardinality of $\SX$ such that there exists such a separation. We also upper-bound the gap between $\pg{c}{\X}{\A;\B}{P}$ and $\pg{ns}{\X}{\A;\B}{P}$ when $\SA$ and $\SB$ have cardinality two and $\SX$ is arbitrary.

\begin{proposition}
\label{prop:binary-input-bounds}
Let $P_\XAB$ be such that $\card{\SA} =  \card{\SB} = 2$. If $\card{\SX} = 2$ then
\begin{align}
    \pg{c}{\X}{\A;\B}{P} = \pg{q}{\X}{\A;\B}{P} = \pg{ns}{\X}{\A;\B}{P}.\label{eq:pg-bin}
\end{align}
If $\card{\SX} > 2$ then
\begin{align}
    \pg{ns}{\X}{\A;\B}{P} \leq \min \set[\Big]{2\pg{c}{\X}{\A;\B}{P}, \pg{c}{\X}{\A;\B}{P} + \frac{1}{8}}.
\end{align}
\end{proposition}

\begin{proof}
To show \cref{eq:pg-bin}, WLOG we assume that $\SA = \SB = \SX = [2]$. By \cref{lm:binary-output-pg}, it suffices to show that
\begin{align}
    \pg{c}{X}{A;B}{P} \geq \max_{f, g} \frac{1}{2} \P[(X, A, B) \sim P_\XAB]{f(A, X), g(B, X)\in [2] \textnormal{ and } f(A, X) - g(B, X) \textnormal{ mod } 2 = AB}, \label{eq:pg-c-ns-bin}
\end{align}
where the maximum is taken over all functions $f, g: [2]\times [2] \to [2]$. Note first that
\begin{align}
    \max_{f, g} \frac{1}{2} \P[(X, A, B) \sim P_\XAB]{f(A, X), g(B, X)\in [2] \textnormal{ and } f(A, X) - g(B, X) \textnormal{ mod } 2 = AB} \leq \frac{1}{2}.
\end{align}
Because $\SX$ is of size two, we also have
\begin{align}
     \pg{c}{\X}{\A;\B}{P}  \geq \max_{x\in\SX} P_{\X}(x) \geq \frac{1}{2}.
\end{align}
Therefore we have \cref{eq:pg-c-ns-bin} as desired.

When $\SX = [d]$ for  $d > 2$, we fix two functions $f, g:[2]\times [d]\to [d]$ such that for all $a, b\in [2]$, $f(a, \cdot): [d] \to [d]$ and $g(b, \cdot): [d] \to [d]$ are bijections. Let $f', g':[2]\times [d] \to [d]$ be such that for all $a, b, x, x'$, we have
\begin{align}
    f(a, x) &= x' \Longleftrightarrow f'(a, x') = x,\\
    g(b, x) &= x' \Longleftrightarrow g'(b, x') = x.
\end{align}
Then,
\begin{align}
    & \P[(X, A, B) \sim P_\XAB]{f(A, X), g(B, X)\in [k] \textnormal{ and } f(A, X) - g(B, X) \textnormal{ mod } k = AB},\\
    &\leq\sum_{i=0,1}\P[(X, A, B) \sim P_\XAB]{f(A, X), g(B, X)\in [k] \textnormal{ and } f(A, X) - g(B, X) \textnormal{ mod } k = i}\\
    &\leq \sum_{i=0, 1} \sum_{j\in[k]} \P[(X, A, B) \sim P_\XAB]{f(A, X) = j, g(B, X) = (j + i) \text{ mod } k}  \\
    &\leq \sum_{i=0, 1} \sum_{j\in[k]} \P[(X, A, B) \sim P_\XAB]{f'(A, j) = X, g'(B, (j + i) \text{ mod } k) = X}\\
    &\leq 2 k \pg{c}{\X}{\A;\B}{P}.
\end{align}
Since $f, g$ are arbitrary, we conclude by \cref{lm:binary-output-pg} that $\pg{ns}{X}{A;B}{P} \leq 2\pg{c}{X}{A;B}{P}$.

Next with re-labeling $a$ and $b$ we can always assume that $\P{AB = 1} \leq 1/4$. Then
\begin{align}
    &\P[(X, A, B) \sim P_\XAB]{f(A, X), g(B, X)\in [k] \textnormal{ and } f(A, X) - g(B, X) \textnormal{ mod } k = AB}\\
    &= \sum_{i=0, 1}\P[(X, A, B) \sim P_\XAB]{f(A, X), g(B, X)\in [k] \textnormal{ and } f(A, X) - g(B, X) \textnormal{ mod } k = i \text{ and } AB = i} \P{AB = i}\\
    &\leq \P[(X, A, B) \sim P_\XAB]{f(A, X), g(B, X)\in [k] \textnormal{ and } f(A,X) - g(B, X) \textnormal{ mod } k = 0} + 1/4.
\end{align}
With the same argument as before,
\begin{align}
    \P[(X, A, B) \sim P_\XAB]{f(A, X), g(B, X)\in [k] \textnormal{ and } f(A,X) - g(B, X) \textnormal{ mod } k = 0} \leq k \pg{c}{\X}{\A;\B}{P}.
\end{align}
Applying \cref{lm:binary-output-pg} again,
\begin{align}
    \pg{ns}{\X}{\A;\B}{P} \leq \sup_{k\geq 2} \pg{c}{\X}{\A;\B}{P} + \frac{1}{4k} \leq \pg{c}{\X}{\A;\B}{P} + \frac{1}{8} ,
\end{align}
as desired.
\end{proof}

\section{Multipartite LSSD is NP-hard} \label{sec:hardness}

In this section we consider the multipartite LSSD problem.
We show in \cref{cor:3-partite-np-hard} that finding an optimal strategy is NP-hard already for three parties with classical inputs.
All games considered in this section are based on probability distributions that corresponds to a uniform distribution over edges of a hypergraph.

\newcommand{\calA}{{\mathcal{A}}}
\newcommand{\calB}{{\mathcal{B}}}
\newcommand{\calG}{{{G}}}
\newcommand{\calV}{{\mathcal{V}}}
\newcommand{\calE}{{\mathcal{E}}}
\newcommand{\calM}{{\mathcal{M}}}

\subsection{Hypergraphs and (partial) matchings} \label{sec:hypergraphs}

A \emph{hypergraph} $\calG$ is a pair $(\SV, \SE)$ where $\SV$ is a set of vertices and $\SE$ is a set of hyperedges, which are non-empty subsets of $\SV$. A \emph{matching} of a hypergraph $\calG = (\SV, \SE) $ is a subset $\SM\subset \SE$ of mutually disjoint hyperedges. We denote by $\nu(\calG)$ the maximum cardinality of a matching of $\calG$. A \emph{fractional matching} of a hypergraph $\calG = (\SV, \SE)$ is a function $g: \SE \to [0, 1]$ such that $\sum_{e\in \SE: v\in e} g(e) \leq 1$ for all $v\in \SV$. We denote by $\nu_f(\calG)$ the maximum of $\sum_{e\in\SE} g(e)$ for all fractional matchings $g$. For any matching $\SM$, $g: e\mapsto \indic{e\in\SM}$ is a fractional matching and therefore we always have $\nu(\calG)\leq \nu_f(\calG)$.

We call a hypergraph $\calG = (\SV, \SE)$ \emph{$r$-partite} if $\SV$ can be partitioned into $r$ parts such that each hyperedge contains precisely one vertex from each part. If we denote the $r$ parts by $\SA_1, \dotsc, \SA_r$, we can characterize a hyperedge $e$ by $(a_1, \dotsc, a_r) \in \SA_1\times \cdots \times \SA_r$ where $a_i$ is the unique vertex in $e\cap \SA_i$. We can thus represent an $r$-partite hypergraph by $(\SA_1, \dotsc, \SA_r, \widetilde{\SE})$ where $\widetilde{\SE} \subset \SA_1\times \cdots \times \SA_r$.

\subsection{Hypergraph games}

For each hpyergraph, we can introduce a probability distribution and a corresponding LSSD game. Note that we need to extend all definitions from \cref{sec:prob-form} from two guessing parties to multi-party guessing, which can be done in a natural way.

\begin{definition}
Let $\calG  = (\SA_1, \dotsc, \SA_r, \SE)$ be an $r$-partite hypergraph. We define a  probability distribution over $\SE \times \SA_1 \times \cdots \times \SA_r$ as
\begin{align}
    P^{\calG}_{\E\A_1\cdots \A_r}(e, a_1, \dotsc, a_r) \eqdef \frac{1}{\card{\SE}} \indic{e = (a_1, \dotsc, a_r)}.
\end{align}
In other words, the random variable $E$ is a uniformly chosen hyperedge of $\calG$ and $A_i$ is the vertex of $E$ in $\SA_i$.
\end{definition}

Our main result of this section relates the optimal guessing probability of the game associated to a hypergraph to its maximum matching.

\begin{theorem}
\label{th:lssd-hypergraph}
For any $r$-partite hypergraph $\calG = (\SA_1, \dotsc, \SA_r, \SE)$,
\begin{align}
    \pg{c }{\E}{\A_1;\dots;\A_r}{P^{\calG}} &= \frac{\nu(\calG)}{\card{\SE}}, \\
    \pg{ns}{\E}{\A_1;\dots;\A_r}{P^{\calG}} &\leq \frac{\nu_f(\calG)}{\card{\SE}}. \label{eq:graph-ns-bound}
\end{align}
\end{theorem}
\begin{proof}
Consider a matching $\SM$ of $\calG$. For a fixed $1\leq i \leq r$ we define $h_i:\SA_i \to \SE$ as follows. Given $a \in \SA_i$, there is at most one $e = (a_1, \dotsc, a_r) \in \SM$ such that $a_i = a$. We set $f_i(a) = e$ if there is such hyperedge $e$ and set $f_i(a)$ to an arbitrary hyperedge otherwise. The probability of winning for this strategy is
\begin{align}
    &\sum_{e, a_1, \dotsc, a_r}P^{\calG}_{\E\A_1\cdots \A_r}(e, a_1, \dotsc, a_r)  \indic{h_1(a_1) = \cdots = h_r(a_r) =  e}\\
    &= \frac{1}{\card{\SE}}\sum_{e, a_1, \dotsc, a_r} \indic{e= (a_1, \dotsc, a_r))}\indic{h_1(a_1) = \cdots = h_r(a_r) =  e}\\
    &= \frac{1}{\card{\SE}}\sum_{e =( a_1, \dotsc, a_r) \in \SE} \indic{h_1(a_1) = \cdots = h_r(a_r) =  e}\\
    &\leq \frac{1}{\card{\SE}}\sum_{e =( a_1, \dotsc, a_r) \in \SM} \indic{h_1(a_1) = \cdots = h_r(a_r) =  e}\\
    &= \frac{\card{\SM}}{\card{\SE}},
\end{align}
which implies that $\pg{c}{\E}{\A_1;\dots;\A_r}{P^{\calG}} \geq \frac{\nu(\calG)}{\card{\SE}}$.

To show the other direction, consider an arbitrary classical strategy described by functions $h_1, \dotsc, h_r$. Define the subset
\begin{align}
    \SM := \set{e = (a_1, \dotsc, a_r) \in \SE: h_1(a_1) = \cdots = h_r(a_r) = e}.
\end{align}
To show that $\SM$ is a matching, let $e = (a_1, \dotsc, a_r)$ and $e'= (a_1', \dotsc, a_r')$ be two distinct hyperedges in $\SM$. Also suppose that $a_i = a_i'$ for some $i$. From the definition of $\SM$, we have $e= h_i(a_i) = h_i(a_i') = e'$ which contradicts the distinctness of $e$ and $e'$. Therefore, $e$ and $e'$ differ in all vertices and $\SM$ is a matching. Next, note that
\begin{align}
    &\sum_{e, a_1, \dotsc, a_r}P^{\calG}_{\E\A_1\cdots \A_r}(e, a_1, \dotsc, a_r)  \indic{h_1(a_1) = \cdots = h_r(a_r) =  e}\\
    &= \frac{1}{\card{\SE}}\sum_{e, a_1, \dotsc, a_r} \indic{e= (a_1, \dotsc, a_r)}\indic{h_1(a_1) = \cdots = h_r(a_r) =  e}\\
    &= \frac{1}{\card{\SE}}\sum_{e =( a_1, \dotsc, a_r) \in \SE} \indic{h_1(a_1) = \cdots = h_r(a_r) =  e}\\
    &= \frac{\card{\SM}}{\card{\SE}}.
\end{align}
Therefore, $\pg{c}{\E}{\A_1;\dots;\A_r}{P^{\calG}} \leq \frac{\nu(\calG)}{\card{\SE}}$.

We now prove \cref{eq:graph-ns-bound}. Let $Q_{\E_1\cdots \E_r|\A_1\cdots \A_r}$ be a no-signaling strategy. For $e = (a_1, \dotsc, a_r) \in \SE$, we define
\begin{align}
g(e) \eqdef Q_{\E_1\cdots \E_r|\A_1\cdots \A_r}\pr{e, \dotsc, e|a_1, \dotsc, a_r}.
\end{align}
We have $g(e) \in [0, 1]$ and for any $a \in \SA_i$
\begin{align}
    \sum_{e = (a_1, \dotsc, a_r) \in \SE: a_i= a} g(e)
    &=  \sum_{e = (a_1, \dotsc, a_r) \in \SE: a_i= a} Q_{\E_1\cdots \E_r|\A_1\cdots \A_r}\pr{e, \dotsc, e|a_1, \dotsc, a_r}\\
    &\leq \sum_{e = (a_1, \dotsc, a_r) \in \SE:a_i = a} \sum_{e_1, \dotsc, e_{i-1}, e_{i+1}, \dotsc, e_{r} } Q_{\E_1\cdots \E_r|\A_1\cdots \A_r}\pr{e_1, \dotsc, e_{i-1}, e, e_{i+1}, \dotsc,  e|a_1, \dotsc, a_r}\\
    &=\sum_{e = (a_1, \dotsc, a_r) \in \SE:a_i = a} Q_{\E_i|\A_1, \dotsc, \A_r}(e|a_1, \dotsc, a_r)\\
    &\stackrel{(a)}{=}\sum_{e = (a_1, \dotsc, a_r) \in \SE:a_i = a} Q_{\E_i|\A_i}(e|a)\\
    &\stackrel{(b)}{\leq} 1,
\end{align}
where $(a)$ follows since $Q_{\E_1\cdots \E_r|\A_1\cdots \A_r}$ is non-signaling and $(b)$ follows since $Q_{\E_i|\A_i}$ is a conditional probability distribution. Therefore, $g$ is a fractional matching. We can upper-bound the probability of winning for the non-signaling strategy $Q_{\E_1\cdots \E_r|\A_1\cdots \A_r}$ as
\begin{align}
    &\sum_{e, a_1, \dotsc, a_r} P^{\calG}_{\E\A_1\cdots \A_r}(e, a_1, \dotsc, a_r) Q_{\E_1\cdots\E_r|\A_1\cdots\A_r}(e, \dotsc, e|a_1, \dotsc, a_r)\\
    &= \frac{1}{\card{\SE}} \sum_{e = (a_1, \dotsc, a_r)} Q_{\E_1\cdots \E_r|\A_1\cdots \A_r}(e, \dotsc, e|a_1, \dotsc, a_r)\\
    &= \frac{1}{\card{\SE}} \sum_{e = (a_1, \dotsc, a_r)} g(e)\\
    &\leq \frac{\nu_f(\calG)}{\card{\SE}},
\end{align}
which completes the proof of \cref{eq:graph-ns-bound}.
\end{proof}

\begin{corollary}
\label{cor:3-partite-np-hard}
For a $3$-partite hypergraph $\calG$, finding $\pg{c}{\E}{\A_1;\A_2;\A_3}{P^{\calG}}$ is an NP-hard problem.
\end{corollary}

\begin{proof}
According to \cref{th:lssd-hypergraph}, finding $\pg{c}{\E}{\A_1;\A_2;\A_3}{P^{\calG}}$ is equivalent to finding the size of the maximum matching in $\calG$, which is NP-hard \cite{Karp1972}.
\end{proof}

\begin{corollary}
Given the assumption $\textnormal{P} \neq \textnormal{NP}$, there exists a $3$-partite hypergraph $\calG$ such that
\begin{align}
    \pg{c}{\E}{\A_1;\A_2;\A_3}{P^{\calG}} < \pg{ns}{\E}{\A_1;\A_2;\A_3}{P^{\calG}}.
\end{align}
\end{corollary}

\begin{proof}
For the sake of contradiction, suppose that for all $3$-partite hypergraphs $\calG$,
\begin{align}
\label{eq:c-equal-ns}
    \pg{c}{\E}{\A_1;\A_2;\A_3}{P^{\calG}} = \pg{ns}{\E}{\A_1;\A_2;\A_3}{P^{\calG}}.
\end{align}
Since $\pg{ns}{\E}{\A_1;\A_2;\A_3}{P^{\calG}}$ can be formulated as a linear program of size polynomial in $\card{\SA_1}\card{\SA_2}\card{\SA_3}$, we can find $\pg{ns}{\E}{\A_1;\A_2;\A_3}{P^{\calG}}$ in polynomial time. Therefore, by our assumption in \cref{eq:c-equal-ns}, we can also find $\pg{c}{\E}{\A_1;\A_2;\A_3}{P^{\calG}}$ in polynomial time, which is in contradiction with \cref{cor:3-partite-np-hard} and the assumption $\textnormal{P} \neq \textnormal{NP}$.
\end{proof}

\begin{corollary}
\label{cor:r-partite-bound}
For any $r$-partite hypergraph $\calG = (\SA_1, \dotsc, \SA_r, \SE)$,
\begin{align}
    \pg{ns}{\E}{\A_1;\dots;\A_r}{P^{\calG}} \leq (r-1) \pg{c}{\E}{\A_1;\dots;\A_r}{P^{\calG}}.
\end{align}
\end{corollary}

\begin{proof}
For any $r$-partite hypergraph $\calG$, we have $\nu_f(\calG) \leq (r-1) \nu(\calG)$ \cite{Furedi_1981}. Combining this with \cref{th:lssd-hypergraph} completes the proof.
\end{proof}

\begin{corollary}
For a bipartite graph $\calG$,
\begin{align}
\label{eq:bipartite-lssd}
    \pg{c}{\E}{\A_1;\A_2}{P^{\calG}} = \pg{q}{\E}{\A_1;\A_2}{P^{\calG}} = \pg{ns}{\E}{\A_1;\A_2}{P^{\calG}}.
\end{align}
\end{corollary}

\begin{proof}
Applying \cref{cor:r-partite-bound} when $r=2$, we have $\pg{ns}{\E}{\A_1;\A_2}{P^{\calG}} \leq \pg{c}{\E}{\A_1;\A_2}{P^{\calG}}  $. On the other hand, $\pg{c}{\E}{\A_1;\A_2}{P^{\calG}} \leq \pg{q}{\E}{\A_1;\A_2}{P^{\calG}} \leq \pg{ns}{\E}{\A_1;\A_2}{P^{\calG}}  $ by definition. Therefore, \cref{eq:bipartite-lssd} holds.
\end{proof}

\section*{Acknowledgements}
We would like to thank Laura Mančinska and Eric Chitambar for useful discussions and for sharing a draft of~\cite{CM21} with us.
CM was supported by an NWO Veni grant (Project No.~VI.Veni.192.159).
MO was supported by an NWO Vidi grant (Project No.~VI.Vidi.192.109).
CS and MT were supported by an NWO Vidi grant (Project No.~639.022.519).

\bibliographystyle{alphaurl}
\newcommand{\etalchar}[1]{$^{#1}$}

\appendix

\section{Proof of \cref{lm:binary-output-pg}}
\label{sec:pf-binary-output-pg}

\pguess*

\begin{proof}
In the classical case, it is enough to consider only deterministic strategies.
They can be described by functions $f: \SA \to \SX$ and $g: \SB \to \SX$ that locally map Alice and Bob's inputs to outputs.
Their success probability is given by
\begin{align}
  \pg{c}{\X}{\A;\B}{P}
  &= \sum_{\substack{x\in\SX\\a\in\SA,b\in\SB}} P_{\XAB}(x,a,b) \indic{f(a) = x} \indic{g(b) = x} \\
  &= \sum_{a,b} P_{\XAB}(f(a),a,b) \indic{f(a) = g(b)}.
  \label{eq:fa}
\end{align}
There are two possibilities: Alice can either ignore her input and always produce a fixed output, or she can take her input into account.

In the first case, $f(0) = f(1) =: s$ and their success probability is
\begin{equation}
  \sum_{a,b} P_{\XAB}(s,a,b) \indic{s = g(b)}.
\end{equation}
It is maximized when Bob also ignores his input and outputs the same fixed value $s$ as Alice, i.e., $g(0) = g(1) = s$.
This results in success probability
\begin{equation}
  \sum_{a,b} P_{\XAB}(s,a,b) = P_\X(s)
  \label{eq:Ps}
\end{equation}
where $s \in \set{0,1}$.
This accounts for the first term in \cref{eq:cprob}.

If Alice does \emph{not} ignore her input then $f(0) \neq f(1)$.
We can assume that neither does Bob, i.e., $g(0) \neq g(1)$.
Indeed, if Bob were to ignore his input, Alice could improve her strategy by outputting the same value as Bob and we would again arrive at \cref{eq:Ps}.
To maximize the success probability in \cref{eq:fa}, the strategies $f$ and $g$ should be coordinated so that $\set{f(0),f(1)} = \set{g(0),g(1)}$ as sets.
In other words, either $f(0) = g(0)$ and $f(1) = g(1)$, or $f(0) = g(1)$ and $f(1) = g(0)$.
These two cases result in success probabilities
\begin{align}
  P_\XAB(f(0),0,0) + P_\XAB(f(1),1,1), &&
  P_\XAB(f(0),0,1) + P_\XAB(f(1),1,0).
\end{align}
Letting $\set{s,t} := \set{f(0),f(1)} \subseteq \SX$ we recover the last two terms in \cref{eq:cprob}.

We now prove \cref{eq:nsprob}. Recall from \cref{eq:pns def} that
\begin{align}
  \pg{ns}{\X}{\A;\B}{P} \eqdef
  \sup_{Q_{\X_A\X_B|\AB}}
  \sum_{\substack{x\in\SX\\a\in\SA,b\in\SB}}
  P_{\XAB}(x,a,b)
  Q_{\X_A\X_B|\AB}(x,x|a,b).
\end{align}
where $Q_{\X_A\X_B|\AB}$ is a conditional probability distribution satisfying the no-signaling conditions in \cref{eq:ns1,eq:ns2}.
Since the objective function and all constraints are linear, an optimal $Q_{\X_A\X_B|\AB}$ is an extreme point of the set of all non-signaling conditional probability distributions.
A \emph{local} extreme point can achieve success probability at most $\pg{c}{\X}{\A;\B}{P}$, corresponding to the first term in \cref{eq:nsprob}.

According to \cite[Theorem~1]{ExtremalPoints}, any \emph{non-local} extreme point of the two-party non-signaling polytope where each party has two inputs and $d$ outputs, is given by $Q^k_{\X_A\X_B|\AB}$ in \cref{eq:Qk}, for some $k \in \set{2, \dotsc, d}$, up to reversible local relabeling.
Intuitively, \cref{eq:Qk} says that we choose $x_B \in [k]$ uniformly at random and set
\begin{equation}
  x_A =
  \begin{cases}
    x_B + 1 \pmod k & \text{if $(a,b) = (1,1)$}, \\
    x_B & \text{otherwise}.
  \end{cases}
  \label{eq:differ on 11}
\end{equation}
A reversible local relabeling means that each party can locally permute their input as well as output values, and the output permutation may depend on the local input value.
The extreme distributions in \cref{eq:Qk} have the property that any local permutation of input values can be achieved by instead locally permuting outputs conditioned on inputs.
For example, the input permutation $a \mapsto 1 - a$ for Alice can be achieved by first negating both variables (i.e., $x_A \mapsto -x_A$ and $x_B \mapsto -x_B$) and then Bob increasing his output by one (i.e., $x_B \mapsto x_B + 1$) whenever $b = 1$.
Indeed, this will cause $x_A = x_B + 1$ whenever $(a,b) = (0,1)$ and $x_A = x_B$ otherwise, see \cref{eq:differ on 11}.

Since we only need to take into account local output permutations that may depend on local inputs, any non-local extreme point of the non-signaling polytope is of the form
\begin{equation}
  \widetilde{Q}^k_{\X_A\X_B|\AB}\of[\big]{x_A,x_B|a,b}
  = Q^k_{\X_A\X_B|\AB}\of[\big]{f(x_A,a),g(x_B,b)|a,b},
\end{equation}
where $Q^k_{\X_A\X_B|\AB}$ is given by \cref{eq:Qk} and $f: \SX \times \SA \to \SX$ and $g: \SX \times \SB \to \SX$ are functions such that $f(\cdot,a), g(\cdot,b): \SX \to \SX$ are permutations for every $a \in \SA$ and $b \in \SB$.
This establishes \cref{eq:nsprob}.
\end{proof}

\section{Constraints on optimal measurements}\label{apx:constraints}

The following proposition shows that any measurement can be replaced by a projective measurement on a larger space.

\begin{proposition}\label{prop:measurement purification}
For any an $n$-outcome measurement
$M = \set{M_1, \dotsc, M_n}$ on $\C^d$,
there is a projective measurement
$\set{\Pi_1, \dotsc, \Pi_n}$ on $\C^d \x \C^n$
and an isometry $U: \C^d \to \C^d \x \C^n$ such that,
for all $i = 1, \dotsc, n$,
\begin{equation}
  M_i = U\ct \Pi_i U.
  \label{eq:purification}
\end{equation}
\end{proposition}

\begin{proof}
Let $U := \sum_{i=1}^n \sqrt{M_i} \x \ket{i}$
and $\Pi_i := \one \x \proj{i}$.
Clearly, each $\Pi_i$ is a projector and $\sum_{i=1}^n \Pi_i = \one$.
\Cref{eq:purification} holds since
\begin{align}
  U\ct \Pi_i U
  &= \of[\bigg]{\sum_{j=1}^n \sqrt{M_j} \x \bra{j}} \Pi_i
     \of[\bigg]{\sum_{k=1}^n \sqrt{M_k} \x \ket{k}} \\
  &= \sum_{j,k=1}^n
     \sqrt{M_j} \one \sqrt{M_k} \x \braket{j}{i} \braket{i}{k} \\
  &= M_i.
\end{align}
Finally, $U$ is an isometry since
$ U \ct U
  = U \ct \of[\big]{\sum_{i=1}^n \Pi_i} U
  = \sum_{i=1}^n U \ct \Pi_i U
  = \sum_{i=1}^n M_i
  = \one$.
\end{proof}

Using the above result, we can show that it suffices to consider only projective measurements when determining the optimal winning probability for quantum strategies assisted by an entangled state of an arbitrarily large dimension.
Our argument is similar to \cite[Lemma~9]{MonogamyGame}.
\begin{corollary}\label{cor:projective}
If $P_\XAB$ is a probability distribution over $\SX \times \SA \times \SB$ then
\begin{align}
  \pg{q}{\X}{\A;\B}{P}
= \sup_{d \geq 1} \pg[d]{q}{\X}{\A;\B}{P}
= \sup_{d \geq 1}
  \sup_{\substack{\Pi: \SA \to \PM{\C^{\da}} \\ \Sigma: \SB \to \PM{\C^{\db}}}}
  \norm[\Big]{\sum_{\substack{x\in\SX\\a\in\SA,b\in\SB}} P_{\XAB}(x,a,b) \Pi_x(a) \x \Sigma_x(b)},
  \label{eq:projective only}
\end{align}
where the last supremum is over collections of projective measurements.
\end{corollary}

\newcommand{\tU}{\widetilde{U}}
\newcommand{\tV}{\widetilde{V}}
\newcommand{\ts}{\widetilde{\sigma}}

\begin{proof}
The first equality in \cref{eq:projective only} is by definition, see \cref{eq:pq def}.
For the second equality, recall from \cref{eq:pq for distribution} that
\begin{equation}
  \pg[d]{q}{\X}{\A;\B}{P}
= \sup_{\substack{M: \SA \to \POVM{\C^{\da}} \\ N: \SB \to \POVM{\C^{\db}}}}
  \norm[\Big]{\sum_{\substack{x\in\SX\\a\in\SA,b\in\SB}} P_{\XAB}(x,a,b) M_x(a) \x N_x(b)}.
\end{equation}
We need to show that, at the cost of increasing the dimension $d$, the optimization here can be restricted to just projective measurements.
For convenience, let
\begin{equation}
  \Omega_\ab := \sum_{\substack{x\in\SX\\a\in\SA,b\in\SB}} P_{\XAB}(x,a,b) M_x(a) \x N_x(b).
  \label{eq:Omega}
\end{equation}
where $M_x^a$ and $N_x^b$ act on registers $\a$ and $\b$ of dimension $d$.

Let us fix a dimension $d \geq 1$ and set $\HA = \C^\SA$ and $\HB = \C^\SB$ as usual.
Using \cref{prop:measurement purification}, we can find collections of isometries $U_a: \C^d \to \C^d \x \HA$ and $V_b: \C^d \to \C^d \x \HB$ and projective measurements $\Pi(a) \in \PM{\C^d \x \HA}$ and $\Sigma(b) \in \PM{\C^d \x \HB}$ such that
\begin{align}
  M_x(a) &= U_a\ct \Pi^a_x U_a, &
  N_x(b) &= V_b\ct \Sigma^b_x V_b,
\end{align}
for all $a \in \SA$, $b \in \SB$, and $x \in \SX$. Then
\begin{equation}
  \Omega_\ab =
  \sum_{\substack{x\in\SX\\a\in\SA,b\in\SB}} P_{\XAB}(x,a,b)
  \of[\big]{U_a \x V_b}\ct
  \of[\big]{\Pi_x(a) \x \Sigma_x(b)}
  \of[\big]{U_a \x V_b}.
  \label{eq:Omega2}
\end{equation}
Let $\ket{\sigma}_\ab \in \C^d \x \C^d$ denote its principal eigenvector.

Let us fix some arbitrary states $\ket{\alpha} \in \HA$ and $\ket{\beta} \in \HB$, and arbitrarily extend the isometries $U_a$ and $V_b$ to unitaries
$\tU_a \in \U{\C^d \x \HA}$ and $\tV_b \in \U{\C^d \x \HB}$ so that
\begin{align}
  U_a &= \tU_a \of[\big]{\one_\a \x \ket{\alpha}_\A}, &
  V_b &= \tV_b \of[\big]{\one_\b \x \ket{\beta}_\B}.
  \label{eq:compatibility}
\end{align}
Furthermore, we promote $\ket{\sigma}_\ab \in \C^d \x \C^d$ to $\ket{\ts}_{\a\A,\b\B} \in (\C^d \x \HA) \x (\C^d \x \HB)$ by defining
\begin{equation}
  \ket{\ts}_{\a\A,\b\B}
 := \ket{\sigma}_\ab \x
    \ket{\alpha}_\A \x
    \ket{\beta}_\B,
\end{equation}
where the registers on the right-hand side should be rearranged accordingly.
Then
\begin{equation}
  \of[\big]{U_A \x V_b}
  \ket{\sigma}_{\ab}
  = \of[\big]{\tU_a \x \tV_b}
  \ket{\ts}_{\a\A,\b\B}
\end{equation}
because of \cref{eq:compatibility}.
Substituting this in \cref{eq:Omega2},
\begin{equation}
  \bra{\sigma} \Omega \ket{\sigma}
  = \bra{\ts}\of[\Bigg]{
      \sum_{\substack{x\in\SX\\a\in\SA,b\in\SB}} P_{\XAB}(x,a,b)
      \of[\big]{\widetilde\Pi_x(a) \x \widetilde\Sigma_x(b)}
    }
    \ket{\ts}\end{equation}
where
$\widetilde{\Pi}_x(a) := \tU_a\ct \Pi_x(a) \tU_a$ and
$\widetilde{\Sigma}_x(b) := \tV_b\ct \Sigma_x(b) \tV_b$
are projectors on $\C^d \x \HA$ and $\C^d \x \HB$.

Hence, we have promoted the original $d$-dimensional strategy to one in dimension $d \max \set{|\SA|,|\SB|}$ that uses only projective measurements and achieves the same success probability.
Since $\pg{q}{\X}{\A;\B}{P}$ in \cref{eq:projective only} is defined as a supremum over all $d \geq 1$, this increase of dimension does not matter.
Hence, we can obtain the optimal quantum value by optimizing only over projectors.
\end{proof}

Intuitively, Alice and Bob should never guess values of $x$ that cannot occur based on their local inputs.
The following result shows that optimal measurements for Alice and Bob's quantum strategies can always be assumed to have this property.

\begin{proposition}\label{prop:zero-povm}
Let $P_\XAB$ be a probability distribution on $\SX \times \SA \times \SB$ and $d \geq 1$ an integer.
The supremum in
\begin{align}
  \pg[d]{q}{\X}{\A;\B}{P}
= \sup_{\substack{M: \SA \to \POVM{\C^{\da}} \\ N: \SB \to \POVM{\C^{\db}}}}
  \norm[\bigg]{\sum_{\substack{x\in\SX\\a\in\SA,b\in\SB}} P_{\XAB}(x,a,b) M_x(a) \x N_x(b)}
  \label{eq:pqd}
\end{align}
is achieved by collections of measurements $M(a) = \set{M_x(a) : x \in \SX}$ and $N(b) = \set{N_x(b) : x \in \SX}$ on $\C^d$ with
\begin{align}
  M_x(a) &= 0 \quad \text{if } P_\XA(x,a) = 0 \text{ and } P_\A(a) > 0, \label{eq:zero-M}\\
  N_x(b) &= 0 \quad \text{if } P_\XB(x,b) = 0 \text{ and } P_\B(b) > 0. \label{eq:zero-N}
\end{align}
In particular, if the supremum can be achieved by projective measurements then it can also be achieved by projective measurements that satisfy \cref{eq:zero-M,eq:zero-N}.
\end{proposition}

\begin{proof}
The set of all measurements on a finite-dimensional complex Euclidean space and with a finite output set $\SX$ is compact.
Since the objective function is continuous, the maximum is achieved by some collections of measurements $M(a)$ and $N(b)$.
We can potentially improve Alice's measurement $M^a$ by absorbing those measurement operators $M_{x'}(a)$ that correspond to pairs $(x',a)$ that never occur into other operators.
More specifically, for each $a \in \SA$ with $P_\A(a) > 0$ there exists some $x_a \in \SX$ with $P_\XA(x_a,a) > 0$, so we can absorb all $M_{x'}(a)$ with $P_\XA(x',a) = 0$ into $M_{x_a}(a)$:
\begin{equation}
  \wtM_x^a \eqdef
  \begin{cases}
    0 & \text{if $P_\XA(x,a) = 0$ and $P_\A(a) > 0$}, \\
    M_x(a) + \sum_{x': P_\XA(x',a) = 0} M_{x'}(a) & \text{if $P_\A(a) > 0$ and $x = x_a$}, \\
    M_x(a) & \text{otherwise}.
  \end{cases}
\end{equation}
We can perform a similar procedure for Bob's measurements $N(b)$ to obtain $\wtN(b)$.
It is clear that all $\wtM(a)$ and $\wtN(b)$ are still measurements, and that they satisfy \cref{eq:zero-M,eq:zero-N}.
In particular, if $M(a)$ are projective measurements then so are $\wtM(a)$.
Moreover,
\begin{equation}
  P_\XAB(x,a,b) \wtM_x(a) \x \wtN_x(b) \succeq
  P_\XAB(x,a,b) M_x(a) \x N_x(b),
  \label{eq:SDP relation}
\end{equation}
for all $x,a,b$.
Indeed, if $P_\XAB(x,a,b) = 0$ then this holds trivially, and if $P_\XAB(x,a,b) > 0$ then $\wtM_x(a) \succeq M_x(a)$ and $\wtN_x(b) \succeq N_x(b)$, so $\wtM_x(a) \x \wtN_x(b) \succeq M_x(a) \x N_x(b)$.
Since \cref{eq:SDP relation} still holds when summig over all $x,a,b$,
\begin{align}
  \norm[\Big]{\sum_{\substack{x\in\SX\\a\in\SA,b\in\SB}} P_\XAB(x,a,b) \wtM_x(a) \x \wtN_x(b)} \geq \norm[\Big]{\sum_{\substack{x\in\SX\\a\in\SA,b\in\SB}} P_\XAB(x,a,b) M_x(a) \x N_x(b)},
\end{align}
as desired.
\end{proof}

\begin{lemma}\label{lem:comm-meas}
Let $P_{\XAB}$ be a joint probability distribution. We fix a quantum strategy consisting of a quantum bi-partite state $\sigma_{\ab}$ with $\Ha = \Hb = \C^d$ and collections of measurement $M:\SA \to \POVM{\C^d}$ and $N:\SB \to \POVM{\C^d}$ with output on $\SX$. Let $M$ be such that $[M(a)_x, M(a')_{x'}]=0$ for all $a, a'\in\SA$ and for all $x,x' \in \SX$. Then,
\begin{align}
    \sum_{\substack{x\in\SX\\a\in\SA,b\in\SB}} P_{\XAB}(x, a,b) \tr\sof[\Big]{\sigma_{\ab}\of[\big]{M_x(a)\otimes N_x(b)}} \leq \pg{c}{\X}{\A;\B}{P}.
\end{align}
\end{lemma}
\begin{proof}
Because  Alice's measurement operators commute, she can jointly perform all measurements for all inputs $a\in \SA$ \emph{before} receiving her input to obtain a collection of random variables $\set{X_a: a\in \SA}$ and use  $X_a$  as her output when her input is $a$. Let $\widetilde{\X}$ denote the register containing all $\set{X_a: a\in \SA}$. Equivalently, Alice and Bob can share $\sigma_{\widetilde{\X}\b}$ in the first place, which is a cq state and therefore separable. Let $\sigma_{\widetilde{\X}\b} = \sum_i p_i \sigma^{(i)}_{\widetilde{\X}} \otimes \sigma^{(i)}_{\b}$ where $\set{p_i}$ is a probability distribution. For any collection of measurements $\widetilde{M}:\SA \to \POVM{\widetilde{\HX}}$,
\begin{align}
    &\sum_{\substack{x\in\SX\\a\in \SA,b\in\SB}} P_{\XAB}(x, a,b) \tr\sof[\Big]{\sigma_{\widetilde{\X}\b}\of[\big]{\widetilde{M}_x(a)\otimes N_x(b)}} \\
    &= \sum_i p_i\sum_{\substack{x\in\SX\\a\in \SA,b\in\SB}} P_{\XAB}(x, a,b) \tr\sof[\Big]{\of[\Big]{\sigma_{\widetilde{\X}}^{(i)}\otimes \sigma_{\b}^{(i)}}\of[\big]{\widetilde{M}_x(a)\otimes N_x(b)}} \\
    &= \sum_i p_i\sum_{\substack{x\in\SX\\a\in \SA,b\in\SB}} P_{\XAB}(x, a,b) \tr\sof[\Big]{\sigma_{\widetilde{\X}}^{(i)} \widetilde{M}_x(a)} \tr\sof[\Big]{\sigma_{\b}^{(i)} N_x(b)}.
\end{align}
Therefore, for each $i$, Alice and Bob can use classical strategies $Q_{\X_A|\A}(x_a|a) \eqdef \tr\sof[\big]{\sigma_{\widetilde{\X}}^{(i)}\widetilde{M}_x(a)} $ and $Q_{\X_B|\B}(x_b|b) \eqdef \tr\sof[\big]{ \sigma_{\b}^{(i)} N_x(b)}$, respectively. Hence,
\begin{align}
    \sum_i p_i\sum_{\substack{x\in\SX\\a\in \SA,b\in\SB}} P_{\XAB}(x, a,b) \tr\sof[\Big]{\sigma_{\widetilde{\X}}^{(i)} \widetilde{M}_x(a)} \tr\sof[\Big]{\sigma_{\b}^{(i)} N_x(b)} \leq \sum_i p_i \pg{c}{\X}{\A;\B}{P} = \pg{c}{\X}{\A;\B}{P},
\end{align}
as desired.
\end{proof}

\section{SOS representation}\label{apx:sos}

\begin{lemma}\label{lem:sos}
For any $t > t_* = \frac{16 + \sqrt{13}}{45}$ and $a,b \in [-1,1]$,
the polynomial
\begin{equation}
  f(t,a,b)
  = t^4 - t^3
  + \frac{32 + (1+a)(1+b)}{100} t^2
  - \frac{16 + 3 (1+a)(1+b)}{500} t
  + \frac{(1+a)(1+b) \of[\big]{4 - (1-a)(1-b)}}{5000}
  \label{eq:f}
\end{equation}
is strictly positive.
\end{lemma}

\begin{proof}
Let us first establish that $f(t,a,b) \geq 0$
for all $t \geq t_*$ and $a,b \in [-1,1]$.
This would be evident if we managed to find a representation of $f$ of the form
\begin{equation}
  f(t,a,b) = v(t,a,b)\tp \of[\Big]{Q_1 + (t-t_*) Q_2 + (1-a^2) Q_3 + (1-b^2) Q_4} v(t,a,b),
  \label{eq:f sos}
\end{equation}
where $Q_i$ are fixed positive semi-definite matrices
and $v(t,a,b)$ is a vector whose entries depend on $t,a,b$
(e.g., are monomials in them).
Generally such ``sums of squares'' representations can be found using semi-definite programming (see Lectures~10--14 of Hamza Fawzi~\cite{Fawzi} or Section~3.4.4 of~\cite{SDPbook}).
In our case this is a semi-definite feasibility problem where the matrices $Q_i$ are subject to $Q_i \succeq 0$ and a set of linear constraints obtained by comparing the coefficients of the polynomials in~\cref{eq:f,eq:f sos}.

We found the following exact solution of this problem:
\begin{align}
  v(a,b,t) = \mx{1 \\ a \\ b \\ ab \\ t \\ t^2}, \quad
  Q_1 =
  \mx{
    \alpha & \beta & \beta & \gamma & \delta & \varepsilon  \\
    \beta & \zeta & \eta & \theta & -\frac{3}{1000} & \frac{1}{200} \\
    \beta & \eta & \zeta & \theta & -\frac{3}{1000} & \frac{1}{200} \\
    \gamma & \theta & \theta & \iota & -\frac{3}{1000} & \frac{1}{200} \\
    \delta & -\frac{3}{1000} & -\frac{3}{1000} & -\frac{3}{1000} & \kappa & -\frac{1}{2} \\
    \varepsilon & \frac{1}{200} & \frac{1}{200} & \frac{1}{200} & -\frac{1}{2} & 1
  },
\end{align}
\begin{align}
  Q_2 &=
  \mx{
    \lambda & 0 & 0 & 0 & 0 & 0 \\
    0 & 0 & 0 & 0 & 0 & 0 \\
    0 & 0 & 0 & 0 & 0 & 0 \\
    0 & 0 & 0 & 0 & 0 & 0 \\
    0 & 0 & 0 & 0 & 0 & 0 \\
    0 & 0 & 0 & 0 & 0 & 0
  }, &
  Q_3 &=
  \mx{
    \mu & 0 & \theta & 0 & 0 & 0 \\
    0 & 0 & 0 & 0 & 0 & 0 \\
    \theta & 0 & \nu & 0 & 0 & 0 \\
    0 & 0 & 0 & 0 & 0 & 0 \\
    0 & 0 & 0 & 0 & 0 & 0 \\
    0 & 0 & 0 & 0 & 0 & 0
  }, &
  Q_4 &=
  \mx{
    \mu & \theta & 0 & 0 & 0 & 0 \\
    \theta & \nu & 0 & 0 & 0 & 0 \\
    0 & 0 & 0 & 0 & 0 & 0 \\
    0 & 0 & 0 & 0 & 0 & 0 \\
    0 & 0 & 0 & 0 & 0 & 0 \\
    0 & 0 & 0 & 0 & 0 & 0
  },
\end{align}
where the values of the missing matrix entries are as follows:
\begin{align}
  \alpha &= \frac{973343 + 240821 \sqrt{13}}{371790000}, &
  \beta &= \frac{33139 - 617 \sqrt{13}}{82620000}, &
  \gamma &= \frac{20 - \sqrt{13}}{45000}, &
  \delta &= - \frac{1721 + 62 \sqrt{13}}{81000}, \\
  \varepsilon &= \frac{25 - 2 \sqrt{13}}{600}, &
  \zeta &= \frac{21592 - 2903 \sqrt{13}}{185895000}, &
  \eta &= \frac{-2 + \sqrt{13}}{45000}, &
  \theta &= \frac{-91 + 617 \sqrt{13}}{82620000}, \\
  \iota &= \frac{-47 + 127 \sqrt{13}}{4590000}, &
  \kappa &= \frac{37+\sqrt{13}}{150}, &
  \lambda &= \frac{91 + 31 \sqrt{13}}{20250}, &
  \mu &= \frac{8203 - 1325 \sqrt{13}}{743580000}, \\ &&&&&&
  \nu &= \frac{871 + 127 \sqrt{13}}{9180000}.
\end{align}
The correctness of this decomposition can be verified by plugging these values into \cref{eq:f sos} and comparing the resulting polynomial with \cref{eq:f}.

To verify that $Q_i$ are positive semi-definite, we can simply compute their eigenvalues.
The non-zero eigenvalues of $Q_1$ are
\begin{equation}
  \begin{array}{c}
    1.255390507\dots \\
    0.020376547\dots \\
    0.000059985\dots \\
    0.000024167\dots \\
    0.000015112\dots
  \end{array}
\end{equation}
The remaining matrices $Q_2, Q_3, Q_4$ have rank one and their only
non-zero eigenvalues are
\begin{align}
  \frac{91 + 31 \sqrt{13}}{20250}, &&
  \frac{39377 + 4481 \sqrt{13}}{371790000}, &&
  \frac{39377 + 4481 \sqrt{13}}{371790000}.
\end{align}

To prove that $f(a,b,t) > 0$ when $t > t_*$, expand \cref{eq:f sos} to obtain
\begin{equation}
  f(t,a,b)
  = v\tp Q_1 v
  + (t-t_*) v\tp Q_2 v
  + (1-a^2) v\tp Q_3 v
  + (1-b^2) v\tp Q_4 v.
\end{equation}
Note that all terms are non-negative when $t \geq t_*$ and $a,b \in [-1,1]$.
Since $\lambda > 0$, the second term
\begin{equation}
  (t-t_*) v\tp Q_2 v = (t-t_*) \lambda
\end{equation}
is strictly positive when $t > t_*$.
\end{proof}

\newcommand{\M}[1]{\href{https://reference.wolfram.com/language/ref/#1.html}{\texttt{#1}}}

The above solution was found using \textit{Mathematica}.
First, we used the \M{SemidefiniteOptimization} function to find an initial solution.
Then, for all sufficiently small matrix entries, we included additional linear constraints that force them to be exactly zero.
This resulted in a preliminary solution with sufficiently many zeroes.
Our hope was to convert this to an exact algebraic solution using the \M{RootApproximant} function.
However, this would work only if the solution is isolated (i.e., cannot be perturbed to other nearby solutions) and of sufficiently high accuracy.
Unfortunately, the built-in \M{SemidefiniteOptimization} function cannot obtain high-accuracy solutions.

To overcome this, we had to rely on the generic \M{NMinimize} and \M{FindMinimum} routines that support \M{WorkingPrecision} option.
However, since they do not support semi-definite constraints, we had to use the preliminary solution to choose a sufficiently simple ansatz matrix $A_i$ and set $Q_i = A_i\tp A_i$.
This automatically guarantees that all $Q_i$ are positive semi-definite.
By further tweaking the ansatz we managed to obtain an isolated solution.

To get an exact algebraic solution, we supplied this isolated numerical solution as an initial point to the \M{FindMinimum} routine and, by increasing the \M{WorkingPrecision} option, dialed up the accuracy to several hundreds of digits.
Finally, applying \M{RootApproximant} to $Q_i$, followed by \M{ToRadicals}, produced the above exact algebraic solution.
\end{document}